\documentclass[a4paper]{article}
\usepackage{mathptmx}
\usepackage[utf8]{inputenc}
\usepackage[T1]{fontenc}
\usepackage[english]{babel}
\usepackage{csquotes}
\usepackage{amsmath}
\usepackage{amsthm}
\usepackage{amsfonts}
\usepackage{xfrac}
\usepackage{mathtools}
\usepackage{authblk}
\usepackage{url}
\usepackage[pdfencoding=auto,psdextra,hidelinks]{hyperref}
\usepackage{bookmark}
\usepackage{float}
\usepackage{booktabs}
\usepackage{subcaption}
\usepackage{tikz}
\usetikzlibrary{calc,positioning,shapes,arrows}

\newcommand*{\email}[1]{\href{mailto:#1}{\nolinkurl{#1}}}

\newcommand{\allones}{\ensuremath{\mathbf{1}}}
\newcommand{\allzeros}{\ensuremath{\mathbf{0}}}

\newcommand{\norm}[1]{\ensuremath{\mathopen\lVert #1 \mathclose\rVert}}

\newcommand{\UnitBall}{\ensuremath{\mathrm{B}(\allzeros,\allones)}}
\newcommand{\Simplex}{\ensuremath{\Delta}}
\newcommand{\CornerSimplex}{\ensuremath{\Delta_\mathrm{c}}}

\newcommand{\ZZ}{\ensuremath{\mathbb{Z}}}
\newcommand{\QQ}{\ensuremath{\mathbb{Q}}}
\newcommand{\RR}{\ensuremath{\mathbb{R}}}

\newcommand{\NP}{\ensuremath{\mathrm{NP}}}

\newcommand{\PSPACE}{\ensuremath{\mathrm{PSPACE}}}

\newcommand{\cETR}{\ensuremath{\exists\RR}}
\newcommand{\PPAD}{\ensuremath{\mathrm{PPAD}}}
\newcommand{\FIXP}{\ensuremath{\mathrm{FIXP}}}

\newcommand{\ETR}{\textsc{ETR}}
\newcommand{\SqrtSum}{\textsc{SqrtSum}}

\newcommand{\PARTITION}{\textsc{Partition}}
\newcommand{\QUAD}{\textsc{Quad}}
\newcommand{\HOMQUAD}{\textsc{HomQuad}}

\newcommand{\Exp}{\operatorname*{E}}

\newcommand{\calG}{\ensuremath{\mathcal{G}}}
\newcommand{\calH}{\ensuremath{\mathcal{H}}}
\newcommand{\calS}{\ensuremath{\mathcal{S}}}
\newcommand{\Nout}{\ensuremath{\operatorname{N}^+}}

\newcommand{\Gvar}{\ensuremath{\calG_{\mathrm{var}}}}
\newcommand{\Gmul}{\ensuremath{\calG_{\mathrm{mul}}}}
\newcommand{\Gpoly}{\ensuremath{\calG_{\mathrm{poly}}}}
\newcommand{\GS}{\ensuremath{\calG(\calS)}}
\newcommand{\Gchance}{\ensuremath{\calG_{\mathrm{chance}}}}
\newcommand{\GexistsNE}{\ensuremath{\calG_{\exists \mathrm{NE}}}}
\newcommand{\GnoNE}{\ensuremath{\calG_{\mathrm{no NE}}}}
\newcommand{\Gsure}{\ensuremath{\calG_{\mathrm{sure}}}}
\newcommand{\Gpart}{\ensuremath{\calG_{\PARTITION}}}

\newcommand{\Reach}{\ensuremath{\mathsf{Reach}}}
\newcommand{\Safe}{\ensuremath{\mathsf{Safe}}}

\newtheorem{theorem}{Theorem}
\newtheorem{proposition}{Proposition}
\newtheorem{lemma}{Lemma}
\newtheorem{corollary}{Corollary}

\newtheorem{remark}{Remark}

\begin{document}

\title{\cETR-Completeness of Stationary Nash Equilibria in Perfect Information Stochastic Games}
\author{Kristoffer Arnsfelt Hansen}
\author{Steffan Christ S{\o}lvsten}
\affil{Aarhus University\\\email{{arnsfelt,soelvsten}@cs.au.dk}}
\maketitle

\begin{abstract}
  We show that the problem of deciding whether in a multi-player perfect
  information recursive game (i.e. a stochastic game with terminal rewards)
  there exists a stationary Nash equilibrium ensuring each player a certain
  payoff is \cETR-complete. Our result holds for acyclic games, where a Nash
  equilibrium may be computed efficiently by backward induction, and even for
  deterministic acyclic games with non-negative terminal rewards. We further
  extend our results to the existence of Nash equilibria where a single player
  is surely winning. Combining our result with known gadget games without any
  stationary Nash equilibrium, we obtain that for cyclic games, just deciding
  existence of any stationary Nash equilibrium is \cETR-complete. This holds
  for reach-a-set games, stay-in-a-set games, and for deterministic recursive
  games.
\end{abstract}

\section{Introduction}
The most common solution concept for noncooperative games is that of a Nash
equilibrium (NE), which was shown by Nash~\cite{AM:Nash51} to be guaranteed to
exist in finite games in strategic form. On the other hand, existence of a NE is
not guaranteed in more general models of games, and one must therefore settle
for weaker solutions. From a computational point of view this leads to the
natural problem of deciding whether a given game admits a NE. Likewise, if a NE
is guaranteed to exist this leads to the natural problem of computing a NE. In
case a NE exists it will generally not be unique, and some NE may be more
desirable than others. For instance, if comparing two different NE, \emph{all}
players may strictly prefer the first NE and we might consider the second NE
undesirable. From a computational point of view this leads to the natural
problem of deciding whether a given game admits a NE in which every player
receives payoff meeting a given \emph{payoff demand}. The computational
complexity of these three basic problems naturally depends heavily on the model
of games under consideration.

In the basic setting of finite games in strategic form, the computational
complexity of these problems is now well understood. The problem of computing a
NE was shown to be $\PPAD$-complete for 2-player games by Daskalakis, Goldberg,
and Papadimitriou~\cite{SICOMP:DaskalakisGP2009} and Chen and
Deng~\cite{FOCS:ChenDeng2006} and $\FIXP$-complete for $m$-player games, when $m
\geq 3$, by Etessami and Yannakakis~\cite{SICOMP:EtessamiY2010}. The problem of
deciding existence of a NE meeting given payoff demands was shown to be
\NP-complete for 2-player games by Gilboa and Zemel~\cite{GEB:GilboaZemel1989}
and \cETR-complete for $m$-player games, when $m \geq 3$, by
Garg~et~al.~\cite{TEAC:GargMVY2018}.

Littman~et~al.~\cite{UAI:LittmanRTZ2006} studied the arguably much simpler case
of two-player perfect information extensive form games, which we shall refer to
simply as \emph{tree games}. Here a NE is guaranteed to exist and may be
computed efficiently by backward induction~\cite{NeumannMorgenstern1947}. In
this way one may in fact always find a pure NE. On the other hand, players are
in general required to make probabilistic choices in order to ensure maximum
possible payoff. While Littman~et~al.\ devise an efficient algorithm for
computing the set of NE payoffs for \emph{deterministic games}, they show that
for two-player games with chance-nodes, it is \NP-hard to decide existence of a
NE meeting given payoff demands. One may for two-player games also prove
\NP-membership of this problem, thereby settling its complexity.

A more general setting where backward induction also show existence and
efficient computation of NE is that of perfect information games that are given
as a directed acyclic graph. We shall refer to these simply as \emph{acyclic
  games}. Here the strategies of the players may in general depend on past
history, but we shall here mainly be interested in the simple case when
strategies just depend on the current node of the graph, i.e.\ stationary
strategies.

Our main result is that for $m$-player perfect information acyclic games, $m
\geq 7$, it is \cETR-complete to decide existence of a stationary NE meeting
given payoff demands. This problem is thus presumably significantly harder for
acyclic games than for tree games. Recently several works have proved
\cETR-completeness for decision problems about NE in multiplayer games, but
these all concerns games in strategic
form~\cite{TOCS:SchaeferS2015,TEAC:GargMVY2018,STACS:BiloM2016,STACS:BiloM17,TCS:Hansen2019,SAGT:BerthelsenH19},
or the even more general models of extensive-form games with perfect recall but
imperfect information~\cite{TCS:Hansen2019} and extensive form games with
imperfect recall~\cite{AAMAS:GimbertPS2020}. In contrast, our results are the
first \cETR-completeness results for \emph{perfect information} games.

Acyclic games form a special case of perfect information recursive games, which
again form a special case of perfect information stochastic games. The
complexity of deciding existence of a NE meeting given payoff demands in
multiplayer stochastic games was first studied systematically by Ummels and
Wojtczak~\cite{LMCS:UmmelsW2011,CONCUR:UmmelsW2011}. Motivated by applications
to verification and synthesis of reactive systems, they study the cases of games
where players have $\omega$-regular objectives and of mean-payoff games, in
addition to the special case of recursive games. Ummels and Wojtczak show that
the problem of existence of a NE meeting given payoff
constraints\footnote{Ummels and Wojtczak consider having both lower bounds
  (i.e.\ demands) and upper bounds on payoffs. Their results however also holds
  with few changes assuming just payoff demands.} is \emph{undecidable} for
10-player recursive games with non-negative terminal rewards or for
deterministic 14-player recursive games. Since then,
Das~et~al.~\cite{TAMC:Das2015} improved this, by showing undecidability of
recursive games with non-negative terminal rewards with just 5 players. In the
more general setting of \emph{concurrent} games,
Bouyer~et~al.~\cite{LIPIcs:Boyer2014} even showed undecidability of the problem
of existence of a NE where a given player is surely winning for deterministic
concurrent 3-players games with reachability objectives.

In order to obtain decidability, Ummels and Wojtczak considered positional and
stationary NE. For existence of stationary NE meeting given payoff constraints,
they prove \NP-hardness for 2-player recursive games with non-negative terminal
rewards and for $n$-player deterministic recursive games (with $n$ being part of
the input), and they prove \SqrtSum-hardness for 4-player recursive games with
non-negative terminal rewards and for 8-player deterministic recursive games. On
the other hand, they show \PSPACE-membership of existence of a NE meeting given
payoff constraint for recursive games, games with common $\omega$-regular
objectives, and mean-payoff games. One may observe that their proofs in fact
give \cETR-membership (cf.\ Section~\ref{SEC:BorelMeanPayoff}).

From our initial \cETR-completeness result we show that deciding existence of a
stationary NE meeting given payoff demands is \cETR-complete also for
deterministic 13-player acyclic games with non-negative terminal rewards. To
prove this we make use of a modified version of a gadget constructed by Ummels
and Wojtczak~\cite{CONCUR:UmmelsW2011} to simulate chance nodes. To use this
modified gadget we rely on the fact, that we have proved \cETR-hardness for
acyclic games. In passing, we also observe that the chance node gadget can be
combined with the \NP-hardness result for tree games of
Littman~et~al.~\cite{UAI:LittmanRTZ2006} to give \NP-hardness for deterministic
tree games.

Combining our results for acyclic games with known gadget games without any
stationary NE, we obtain that for cyclic games, just deciding existence of any
stationary NE is \cETR-complete. This holds for reach-a-set games,
stay-in-a-set games, and for deterministic recursive games. Ummels previously
proved \NP-hardness and \SqrtSum-hardness for deciding existence of any
stationary NE in reach-a-set games~\cite[Corollary~4.9]{Ummels2011}.
The gadgets used for the last two constructions were only constructed recently
and to use them we again rely on the fact that we have proved \cETR-hardness
for acyclic games.

\section{Preliminaries}
For a finite set $S$, let $\Delta(S)$ denote the set of probability
distributions on $S$. Denote by $\Simplex^n \subseteq \RR^{n+1}$ the standard
$n$-simplex $\{x \in \RR^{n+1} \mid x\geq 0 \wedge \sum_{i=1}^{n+1} x_i = 1\}$.
We may then identify $\Simplex^n$ and $\Delta(\{1,\dots,n+1\})$ in the natural
way. Denote by $\CornerSimplex^n \subseteq \RR^n$ the standard corner
$n$-simplex $\{x \in \RR^n \mid x\geq 0 \wedge \sum_{i=1}^n x_i \leq 1\}$.

We next define the types of games, payoffs, and equilibria we consider in this
paper. Striving for a uniform exposition we modify common definitions in slight
and non-essential ways.

\subsection{Perfect Information Stochastic Games}
An $m$-player perfect information stochastic game $G$ is given by a directed
graph (digraph) $D=(V,A)$. For $u\in V$ denote by $\Nout(u) = \{v \in V \mid
(u,v)\in A\}$ the out-neigh\-bor\-hood of $u$. Let $T = \{u \in V \mid
\Nout(u)=\emptyset\}$ denote the set of sink nodes of $D$, also called the
\emph{terminals}. The non-terminal nodes are partitioned into disjoint sets $V
\setminus T = V_0 \cup V_1 \cup \dots \cup V_m$, where $V_0$ is the set of
\emph{chance nodes} and $V_i$ is the set of \emph{Player~$i$ nodes}, when $i\geq
1$. To each $v \in V_0$ is assigned a probability distribution $\pi_v \in
\Delta(\Nout(v))$. We say the game $G$ is \emph{deterministic} if
$V_0=\emptyset$.

We fix an initial node $u_0 \in V$ from which play proceeds in rounds. A history
of play is an infinite sequence $(u_k)_{k\geq 0}$ such that $(u_k,u_{k+1})\in A$
when $u_k \notin T$ and $u_{k+1}=u_k$ when $u_k \in T$. Let $\calH_\infty$
denote the set of all such histories. A finite history is a prefix of a history
of play. For $i\geq 0$ and $v \in V_i$, let $\calH_{i,v}$ denote the set of
finite histories $(u_k)_{k=0}^{K}$ ending in node $u_K=v$. For $i\geq 0$, let
$\calH_i = \cup_{v \in V_i} \calH_{i,v}$ denote the finite histories ending in a
node in $V_i$, and finally let $\calH = \cup_{i\geq 0} \calH_i$ denote the set
of all finite histories. If some prefix of a play is contained in $\calH_{i,v}$
for some $i$ and $v \in V_i$ we say that the play \emph{reaches}~$v$. A finite
history $h=(u_k)_{k=0}^K \in \calH$ defines a \emph{subgame} $G[h]$ of $G$ with
$u_K$ being the initial node of $G[h]$, play proceeding from $u_K$ in rounds
extending $h$.

\subsubsection{Strategies and Equilibria}
A strategy $\tau_i$ for Player~$i$ assigns to each $h \in \calH_{i,v}$ a
probability distribution $\tau_i(h) \in \Delta(\Nout(v))$, viewed as a function
$\Nout(v) \rightarrow [0,1]$. The strategy $\tau_i$ is \emph{stationary} if
$\tau_i(h)=\tau_i(h')$ for every $h,h' \in \calH_{i,v}$ and every $v \in V_i$,
i.e.\ when $\tau_i$ only depends on $v$. The strategy $\tau_i$ is \emph{pure} if
$\tau_i(h)$ is a single-point distribution for every $h \in \calH_i$. A
\emph{positional} strategy is a strategy that is simultaneously pure and
stationary.

A strategy profile $\tau=(\tau_1,\dots,\tau_m)$ consists of a strategy for each
player. The strategy profile is stationary, pure, or positional if all of its
strategies are stationary, pure, or positional, respectively. The set of plays
that extend a given finite history $h=(u_k)_{k=0}^K$ is called a cylinder set.
The total probability of these plays is given by the product $\prod_{k=0}^{K-1}
p_k (u_{k+1})$ where $p_k = \tau_i(u_0,\dots,u_{k})$ when $u_k\in V_i$ for some
$i \geq 1$ and where $p_k = \pi_{u_k}$ when $u_k \in V_0$. By Carathéodory's
extension theorem this defines a unique probability measure on the Borel
$\sigma$-algebra generated by the cylinders sets. Assume now that each
Player~$i$ is equipped with a bounded Borel measurable utility function $u_i:
\calH_\infty \rightarrow \RR$. Let $u: \calH_\infty \rightarrow \RR^m$ denote
the vector function of utilities $u(h)=(u_1(h), \dots, u_m(h))$. Given a
strategy profile $\tau$, the expected payoff $U_i(x)$ for Player~$i$ is given by
$U_i(\tau)=\Exp_\tau[u_i(h)]$. We let $U(\tau)=(U_1(\tau),\dots,U_m(\tau))$
denote the payoff profile of $\tau$.

Given a strategy profile $\tau$ we let
$\tau_{-i}=(\tau_1,\dots,\tau_{i-1},\tau_{i+1},\dots,\tau_m)$ denote the
strategy profile of all players except Player~$i$. Given a strategy $\tau'_i$
for Player~$i$, we let $(\tau_{-i};\tau'_i)$ denote the strategy profile
$(\tau_1,\dots,\tau_{i-1},\tau'_i,\tau_{i+1},\dots,\tau_m)$. We also denote
$(\tau_{-i};\tau'_i)$ by $\tau\setminus \tau'_i$. We say that $\tau'_i$ is a
\emph{best reply} for Player~$i$ to $\tau$ if $u_i(\tau \setminus \tau'_i) \geq
u_i(\tau \setminus \tau''_i)$ for all strategies $\tau''_i$ of Player~$i$. We
say that $\tau$ is a \emph{Nash equilibrium} (NE) if $\tau_i$ is a best reply to
$\tau$ for every Player~$i$.

Any finite history $h \in \calH$ induces a \emph{conditional strategy}
$\tau_i[h]$ in the subgame $G[h]$ from a strategy $\tau_i$ of Player~$i$. We say
that $\tau=(\tau_1,\dots,\tau_m)$ is a \emph{subgame perfect equilibrium} (SPE)
if the conditional strategy profile $\tau[h]=(\tau_1[h],\dots,\tau_m[m])$ is a
NE in $G(h)$, for every $h \in \calH$.

\subsubsection{Utility Functions}
We shall consider several different types of utility functions which in turn
gives rise to different classes of games. In a \emph{recursive
  game}~\cite{AMS:Everett1957} only plays that reach a terminal are assigned
non-zero utility. We may thus view the utility functions as functions $u_i : T
\rightarrow \RR$, also known as \emph{terminal rewards}. Recursive games where
all terminal payoffs are non-negative or non-positive are respectively called
non-negative recursive games and non-positive recursive games. If we normalize
the utility functions to take values in the range $[-1,1]$, every terminal
reward vector $u(v)$, for $v \in T$, can be written as a convex combination
$\sum_{i=1}^k \alpha_k p_k$ of vectors $p_k \in \{-1,0,1\}^m$. By replacing
terminal nodes with payoff $u(v)$ with an additional chance node going to a
terminal with payoff $p_k$ with probability $\alpha_k$, we transform a recursive
game into an equivalent recursive game with terminal reward vectors from the set
$\{-1,0,1\}^m$.

In a \emph{mean-payoff game}~\cite{AMS:Gillette1957,IJGT:EhrenfeuchtM1979},
Player~$i$ is given a reward function $r_i : V \rightarrow \RR$ and the utility
assigned to a play $h=(u_k)_{k\geq 0}$ is $u_i(h)=\liminf\limits_{K \rightarrow
  \infty} \tfrac{1}{K}\sum_{k=0}^{K-1} r_i(u_k)$, for all~$i$. Note that a
recursive game is a special case of a mean-payoff game, where all non-terminal
nodes are given reward $0$.

Utility functions that are indicator functions of Borel sets of plays are called
objectives. For convenience we simply identify the objective with its defining
set of plays. For $S \subseteq V$, the \emph{reachability} objective $\Reach(S)$
is the set of plays that reach a node in $S$ and the \emph{safety} objective
$\Safe(S)$ is the set of plays that only reach nodes in $S$. Games in which all
players have reachability objectives are called \emph{reach-a-set
  games}~\cite{CSL:ChatterjeeMJ2004} and games in which all players have safety
objectives are called \emph{stay-in-a-set games}~\cite{IJGT:SecchiS2002}. We say
that $\Reach(S)$ is a terminal reachability objective if $S \subseteq T$ and
similarly that $\Safe(S)$ is a terminal safety objective if $V \setminus T
\subseteq S$. Note that a reach-a-set game with terminal reachability objectives
is equivalent to a recursive game with terminal rewards from the set $\{0,1\}$.
Likewise, a stay-in-a-set game with terminal safety objectives is equivalent to
a recursive game with terminal rewards from the set $\{-1,0\}$. Other objectives
of interest are the standard $\omega$-regular objectives of Büchi, co-Büchi,
Parity, Streett, Rabin, Muller objectives, see e.g.~\cite{LMCS:UmmelsW2011} for
definitions. These objectives all generalize \emph{terminal} reachability and
safety objectives.

\subsubsection{Games on Trees and DAGs}
When the digraph $D$ of a given perfect information stochastic game $G$ is
\emph{acyclic} we refer to $G$ as an \emph{acyclic game}. Likewise, when $D$ is
a tree we refer to $G$ as a \emph{tree game}. A tree game is in particular an
acyclic game.

In an acyclic game we have that every play reaches a terminal. For a general
acyclic game there may be multiple plays reaching the same terminal, but for
a tree game there is a unique play reaching each specific terminal. Thus for a
tree game we may view the utility functions simply as terminal payoffs. This
also means that tree games correspond exactly to perfect information extensive
form games. The method of backward induction~\cite{NeumannMorgenstern1947} shows
existence of a (pure) SPE for any terminal payoff acyclic game, and by
considering the unfolding of an acyclic game into a tree game, also a SPE for
any acyclic game.

\subsection{The Existential Theory of the Reals}
The existential theory of the reals \ETR\ is the set of all true sentences of
the form $\exists x_1, \dots, x_n \in \RR : \varphi(x_1, \dots, x_n)$, where
$\varphi$ is a quantifier-free Boolean formula of inequalities and equalities of
polynomials with integer coefficients. Schaefer and
Štefankovič~\cite{TOCS:SchaeferS2015} defined the complexity class \cETR\ as the
closure of \ETR\ under polynomial time many-one reductions. Alternatively,
\cETR\ is equal to the constant-free Boolean part of the class
$\NP_\RR$~\cite{FOCM:BurgisserC2009}, which is the analogue class to $\NP$ in
the Blum-Shub-Smale model of computation~\cite{BAMS:BlumSS1989}. Clearly $\NP
\subseteq \cETR$ and from the decision procedure by Canny~\cite{JACM:Canny1988}
we have that $\cETR \subseteq$ \PSPACE.

A fundamental complete problem for \cETR\ is the problem $\QUAD$ of deciding
whether a system $\calS$ of quadratic polynomials in $n$ variables with integer
coefficients has a solution in $\RR^n$~\cite{BAMS:BlumSS1989}.
Schaefer~\cite{Schaefer13} proved that the similar problem $\QUAD(\UnitBall)$ of
deciding whether the system $\calS$ has a solution in the unit ball is also
\cETR-complete. Analogously one can prove (cf.~\cite{TCS:Hansen2019}) that the
problem $\QUAD(\CornerSimplex)$ of deciding whether the system $\calS$ has a
solution in the corner simplex $\CornerSimplex^n$ is \cETR-complete.

Define $\HOMQUAD(\Simplex)$ as the problem of deciding whether a system $\calS'$
of homogeneous quadratic polynomials in $n$ variables with integer coefficients
has a solution in the unit simplex $\Simplex^{n-1}$. This problem will form the
basis of our \cETR-hardness results.
\begin{proposition}
  $\HOMQUAD(\Simplex)$ is \cETR-complete.
\label{PROP:HomQUAD}
\end{proposition}
\begin{proof}
  Membership of \cETR\ is straightforward. To obtain \cETR-hardness we reduce
  from $\QUAD(\CornerSimplex)$. Suppose $\calS$ is a system of quadratic
  equations in $n-1$ variables $x_1,\dots,x_{n-1}$. Introduce the slack variable
  $x_n = 1-\sum_{i=1}^{n-1}x_i$. We may then homogenize each polynomial of
  $\calS$ forming the set of homogeneous quadratic polynomials $\calS'$,
  replacing constant terms of the form $a$ by $\sum_{i=1}^n\sum_{j=1}^n ax_ix_j$
  and degree~1 terms of the form $ax_i$ by $\sum_{j=1}^n ax_ix_j$. Solutions of
  $\calS$ in $\CornerSimplex^{n-1}$ then correspond exactly to solutions of
  $\calS'$ in $\Simplex^n$, by either introducing or dropping the slack variable
  $x_n$.
\end{proof}

\section{\cETR-Completeness of Stationary NE} \label{SEC:ETR-Completeness}
Consider an $m$-player game $G$ and let $L \in \RR^m$ be a vector of
\emph{payoff demands}. We say that a strategy profile $\tau$ satisfies the
payoff demands $L$ if $U(\tau)\geq L$ (with component-wise comparison).

Our main result is a precise characterization of the complexity of deciding
existence of stationary NE in perfect information recursive games satisfying
given payoff demands.

\begin{theorem} \label{THM:RecursiveGamesETR} It is \cETR-complete to decide
  whether for a given $m$-player recursive game $G$ and payoff demands
  $L\in\RR^m$ there exists a stationary NE $\tau$ with $U(\tau)\geq L$. The
  problem is \cETR-complete even for acyclic 7-player recursive games with
  non-negative rewards. The same result holds for the analogous problem for
  stationary SPE.
\end{theorem}
Membership of \cETR\ follows by expressing that $\tau$ is a stationary NE (SPE)
satisfying the given payoff demands by an existential first-order formula over
the reals. This is done by expressing for all $i$ that $\tau_i$ is an optimal
solution of the Markov Decision Process (MDP) for Player~$i$ that results from
fixing the strategies of the other players according to $\tau_{-i}$. Ummels and
Wojtczak give a detailed proof for the (more general) case of mean-payoff
games~\cite[Theorem~7]{CONCUR:UmmelsW2011} (see the full version of the
paper~\cite{arXiv:UmmelsW2011} for the actual proof). We return to this in
Section~\ref{SEC:BorelMeanPayoff}.

Our proof of \cETR-hardness is by reduction from the problem
$\HOMQUAD(\Simplex)$ and involves several gadget games that we describe next. In
the following let $\calS$ be a system of homogeneous quadratic polynomials
$q_1(x),\dots,q_\ell(x)$ in variables $x=(x_1,\dots,x_n)$. We write $q_k(x) =
\sum_{i=1}^n \sum_{j=1}^n a_{ij}^{k} x_i x_j$ for $k = 1,\dots,\ell$, and assume
that coefficients are scaled to be rational numbers in the interval $[-1,1]$.
That is $a_{ij}^k \in \mathbb{Q}$ and $-1 \leq a_{ij}^k \leq 1$, for all
$i,j,k$.

\begin{remark}
  For clarity, drawings of the many gadget games are provided in accompanying
  figures. Chance nodes $v \in V_0$ are diamond-shaped with out-going arcs
  labelled by the values of $\pi_v$. Nodes $v \in V_i$ controlled by Player~$i$
  are circular nodes labelled with $i$ above and unlabelled out-going arcs.
  Nodes themselves may also contain labels, though these labels are only used to
  refer to the specific nodes inside the proofs.
\end{remark}
The first gadget is the variable selection game $\Gvar$ shown in
Figure~\ref{fig:Gvar}. An initial chance node leads to Player~$1$ nodes
$v_1,\dots,v_n$, each chosen with probability~$\tfrac{1}{n}$. In node $v_i$,
Player~1 makes a binary choice between either giving payoff~1 to Player~2 and
Player~4 or to Player~3 and Player~5 and all other players payoff~0. We let
$x_i$ denote the probability of the former choice, and let $x=(x_1,\dots,x_n)$.
Since $0\leq x_i\leq 1$, it follows that $x\geq 0$ and $\norm{x}_1\leq n$.

\begin{figure}[ht]
  \centering

  \begin{subfigure}{0.45\linewidth}
    \centering
    \begin{tikzpicture}[shorten >= 1, node distance = 2cm, on grid, minimum size = 0.8cm]
  \node[shape=circle, draw=black, label=above:$1$, label=left:$\rightarrow$] (ni_1) {$v_i$};
  \node[draw=none] (gi_1) [above right=0.5cm and 3cm of ni_1] {$(0,1,0,1,0,0,0)$};
  \node[draw=none] (gi_2) [below right=0.5cm and 3cm of ni_1] {$(0,0,1,0,1,0,0)$};
  \path[->]
     (ni_1) edge [bend left=8] node [above] {\small \color{gray} $x_i$} (gi_1)
            edge [bend right=8] node [below] {\small \color{gray} $1-x_i$} (gi_2)
   ;
\end{tikzpicture}
    \caption{The nodes $v_i$ of $\Gvar$}
  \end{subfigure}
  \begin{subfigure}{0.45\linewidth}
    \centering
    \begin{tikzpicture}[shorten >= 1, node distance = 2cm, on grid, minimum size = 0.8cm]
  \node[shape=diamond, draw=black, label=left:$\rightarrow$] (e1) {};
  \node[draw=none,rectangle] (vi) [right=2cm of e1] {};
  \node[draw=none,rectangle] (v1) [above=.8cm of vi] {$v_1$};
  \node[draw=none,rectangle] (vn) [below=.8cm of vi] {$v_n$};
  \path[->]
     (e1) edge [bend left=10] node [above left] {$\frac{1}{n}$} (v1)
          edge [loosely dashed] (vi)
          edge [bend right=10] node [below left] {$\frac{1}{n}$} (vn)
          
   ;
   \path (v1) edge node [black, opacity=1, sloped] {\dots} (vn) [opacity=0];
\end{tikzpicture}

    \caption{The game $\Gvar$}
  \end{subfigure}

  \caption{The variable selection game $\Gvar$.}
  \label{fig:Gvar}
\end{figure}

The payoff analysis of $\Gvar$ is straightforward.
\begin{lemma} \label{LEM:G_x-payoff}
  The payoff profile of the subgame of $\Gvar$ starting from node $v_i$ is equal
  to $(0,x_i,1-x_i,x_i,1-x_i,0,0)$, for $i=1,\dots,n$. The payoff profile of the
  game $\Gvar$ itself is of the form
  \begin{equation*}
    \left(
      0,
      \tfrac{1}{n}\norm{x}_1,
      1-\tfrac{1}{n}\norm{x}_1,
      \tfrac{1}{n}\norm{x}_1,
      1-\tfrac{1}{n}\norm{x}_1,
      0,
      0
    \right)
    \enspace .
  \end{equation*}
\end{lemma}
\noindent We eventually want to enforce that $x \in \Simplex^{n-1}$ by payoff
demands. Note that this can be obtained locally in $\Gvar$ by payoff demands
$\tfrac{1}{n}$ for Player~2 and $\tfrac{n-1}{n}$ for Player~3.

The second gadget is the multiplication game $\Gmul(i,j,\alpha)$, defined for $1
\leq i,j \leq n$ and $\alpha \in [0,1]$ and shown in Figure~\ref{fig:Gmul}. Note
that it connects to nodes $v_i$ and $v_j$ of $\Gvar$. By
Lemma~\ref{LEM:G_x-payoff} these may be viewed as terminal nodes with reward
vectors $(0,x_i,1-x_i,x_i,1-x_i,0,0)$ and $(0,x_j,1-x_j,x_j,1-x_j,0,0)$, and
we shall do so in the analysis in order to be able to analyze
$\Gmul(i,j,\alpha)$ separately.

First Player~2 and Player~3 are able to threat to leave to node~$v_i$. Otherwise
Player~1 is given a binary choice: either continue or give Player~1 and Player~3
reward~1. We denote by $x'_i$ the probability of the former choice. If Player~1
continues, Player~4 and Player~5 are able to threat to leave to node~$v_j$.
Otherwise Player~1 is given a binary choice between terminal reward vectors
$(1,1,0,1,0,\alpha,1-\alpha)$ and $(1,1,0,0,1,0,0)$. We denote by $x'_j$ the
probability the former choice.
\begin{figure}[ht]
  \centering

  \begin{tikzpicture}[shorten >= 1, node distance = 2cm, on grid, minimum size = 0.8cm]
  \node[shape=circle, draw=black]
      (threat1) {$v_i$};
  \node[shape=circle, draw=black, label=above:$2$, label=left:$\rightarrow$]
      (threat1_1) [above left=1.4cm and 0.6cm of threat1] {$w_1$};
  \node[shape=circle, draw=black, label=above:$3$]
      (threat1_2) [above right=1.4cm and 0.6cm of threat1] {$w_2$};
  \node[shape=circle, draw=black, label=above:$1$]
      (repl1) [right=1.7cm of threat1_2] {$w_3$};
  \node[draw=none]
      (repl1_terminal) [below=1.4cm of repl1] {$(1,0,1,0,0,0,0)$};
  \node[shape=circle, draw=black, label=above:$4$]
      (threat2_1) [right=1.7cm of repl1] {$w_4$};
  \node[shape=circle, draw=black, label=above:$5$]
      (threat2_2) [right=1.2 of threat2_1] {$w_5$};
  \node[shape=circle, draw=black]
      (threat2) [below right=1.4cm and 0.6cm of threat2_1] {$v_j$};
  \node[shape=circle, draw=black, label=above:$1$]
      (repl2) [right=1.7cm of threat2_2] {$w_6$};
  \node[draw=none] (repl2_terminal1)
      [above right=1.5cm and 1.5cm of repl2] {$(1,1,0,1,0,\alpha,1-\alpha)$};
  \node[draw=none] (repl2_terminal2)
      [below right=1.5cm and 1.5cm of repl2] {$(1,1,0,0,1,0,0)$};
  \path[->]
    (threat1_1) edge [bend left=5] (threat1)
    (threat1_1) edge (threat1_2)
    (threat1_2) edge [bend right=5] (threat1)
    (threat1_2) edge (repl1)
    (repl1) edge node [above=-0.15cm] {\small \color{gray} $x_i'$} (threat2_1)
            edge node [right] {\small \color{gray} $1-x_i'$} (repl1_terminal)
    (threat2_1) edge (threat2_2)
    (threat2_1) edge [bend left=5] (threat2)
    (threat2_2) edge [bend right=5] (threat2)
    (threat2_2) edge (repl2)
    (repl2) edge [bend right=10]
        node [right=.2cm] {\small \color{gray} $x_j'$} (repl2_terminal1)
    (repl2) edge [bend left=10]
        node [right=.2cm] {\small \color{gray} $1 - x_j'$} (repl2_terminal2)
  ;
\end{tikzpicture}


  \caption{The multiplication game $\Gmul(i,j,\alpha)$.}
  \label{fig:Gmul}
\end{figure}
\begin{lemma} \label{LEM:G_mult-NE}
  Any NE payoff profile of $\Gmul(i,j,\alpha)$ in which
  Player~1 receives payoff~1 is of the form
  \begin{equation*}
    \left(
      1,
      x_i,
      1-x_i,
      x_ix_j,
      x_i(1-x_j),
      \alpha x_ix_j,
      (1-\alpha)x_ix_j
    \right)
    \enspace .
  \end{equation*}
\end{lemma}
\begin{proof}
  For Player~1 to receive payoff~1, neither of Player~2, 3, 4, or~5 execute
  their threats to leave to $v_i$ or $v_j$ with positive probability.
  Conditioned on play reaching node $w_3$, Player~2 and Player~3 receives payoff
  $x'_i$ and $1-x'_i$, respectively. Thus, unless $x'_i=x_i$, either Player~2 or
  Player~3 would gain by leaving to $v_i$ in node $w_1$ or $w_2$. Similarly,
  conditioned on play reaching node $w_6$, Player~4 and Player~5 receive payoff
  $x'_j$ and $1-x'_j$, respectively. Thus, unless $x'_j=x_j$, either Player~4 or
  Player~5 would gain by leaving to $v_j$ in node $w_4$ or $w_5$. It follows
  that the payoff profile is as claimed.
\end{proof}

The third gadget is the polynomial evaluation game $\Gpoly(k)$ defined by the
polynomial $q_k(x)$ and shown in Figure~\ref{fig:Gpol}. First Player~6 and
Player~7 are in turn able to threat to leave to a terminal giving payoff
$1/(2n^2)$ (and all other players payoff~0). Otherwise a chance node leads to
the game $\Gmul(i,j,(1+a^k_{ij})/2)$, with probability $1/n^2$, for
$i,j=1,\dots,n$.

\begin{figure}[ht]
  \centering
    \begin{tikzpicture}[shorten >= 1, node distance = 2cm, on grid, minimum size = 0.8cm]
    \node[shape=circle, draw=black, label=above:$6$, label=left:$\rightarrow$] (threat1) {};
    \node[draw=none] (threat1_t) [below=1.3cm of threat1] {$(0,0,0,0,0,\frac{1}{2 n^2},0)$};
    \node[shape=circle, draw=black, label=above:$7$] (threat2) [right=3cm of threat1] {};
    \node[draw=none] (threat2_t) [below=1.3cm of threat2] {$(0,0,0,0,0,0,\frac{1}{2n^2})$};
    \node[shape=diamond, draw=black] (chance_node) [right=2cm of threat2] {};
    \node[draw=none] (g1_1) [above right=1.2cm and 3cm of chance_node] {$\Gmul(1,1,\frac{1 + a_{1,1}^k}{2})$};
    \node[draw=none] (gi_j) [above right=0cm and 3cm of chance_node] {$\Gmul(i,j,\frac{1 + a_{i,j}^k}{2})$};
    \node[draw=none] (gn_n) [below right=1.2cm and 3cm of chance_node] {$\Gmul(n,n,\frac{1 + a_{n,n}^k}{2})$};
    \node[draw=none] (gi) [below right=0.6cm and 2cm of chance_node] {};
    \node[draw=none] (gj) [above right=0.6cm and 2cm of chance_node] {};
    \path[->]
      (threat1) edge (threat2)
      (threat1) edge (threat1_t)
      (threat2) edge (threat2_t)
      (threat2) edge (chance_node)
      (chance_node) edge [bend left=20] node [left] {$\frac{1}{n^2}$} (g1_1)
      (chance_node) edge (gi_j)
      (chance_node) edge [bend right=25] node [left] {$\frac{1}{n^2}$} (gn_n)
    ;
    \path[->]
      (chance_node) edge [bend right=10] (gi)
      (chance_node) edge [bend left=10] (gj)[loosely dashed]
    ;
    \path[-]
      (g1_1) edge (gi_j)
      (gi_j) edge (gn_n) [loosely dotted, thick]
    ;
  \end{tikzpicture}
  
  \caption{The polynomial evaluation game $\Gpoly(k)$}
  \label{fig:Gpol}
\end{figure}

The analysis of $\Gpoly(k)$ follows by using Lemma~\ref{LEM:G_mult-NE}.
\begin{lemma} \label{LEM:G_poly-NE}
  Any NE payoff profile of $\Gpoly(k)$ in which Player~1 receive payoff~1 is of
  the form
  \begin{equation*}
    \scalebox{0.9}[1]{$
    \left(
      1,
      \frac{1}{n}\norm{x}_1,
      1-\frac{1}{n}\norm{x}_1,
      (\frac{1}{n}\norm{x}_1)^2,
      \frac{1}{n}\norm{x}_1(1-\frac{1}{n}\norm{x}_1),
      \frac{1}{2n^2}(\norm{x}_1^2+q_k(x)),
      \frac{1}{2n^2}(\norm{x}_1^2-q_k(x))
    \right)
    \enspace .
    $}
  \end{equation*}
\end{lemma}
\begin{proof}
  For Player~1 to receive payoff~1, neither Player~6 nor Player~7 execute their
  threats to leave directly to the terminal nodes. Likewise, Player~1 must
  receive payoff~1 in each of the games $\Gmul(i,j,(1+a_{ij}^k)/2)$, each of
  which by Lemma~\ref{LEM:G_mult-NE} then has the payoff profile
  $(1,x_i,1-x_i,x_ix_j,x_i(1-x_j),(1+a_{ij}^k)x_ix_j/2,(1-a_{ij}^k)x_ix_j/2)$.
  Taking the average of this over all pairs $i,j \in \{1,\dots,n\}$ is easily
  seen to yield the claimed payoff vector. For instance, the payoff of Player~6
  is equal to
  \begin{align*}
    \frac{1}{n^2}\sum_{i=1}^n\sum_{j=1}^n \left(\frac{1+a_{ij}^k}{2}\right)x_ix_j
    &= \frac{1}{2n^2}
      \left(
        \left(\sum_{i=1}^n x_i\right)
        \left(\sum_{j=1}^n x_j\right) +
        \left(\sum_{i=1}^n\sum_{j=1}^n a_{ij}^k x_ix_j\right)
      \right)
    \\ &= \frac{1}{2n^2}(\norm{x}_1^2+q_k(x)) \enspace .
  \end{align*}
\end{proof}
\begin{corollary} \label{COR:G_poly-NE}
  In a NE of $\Gpoly$ where $\norm{x}_1=1$ and Player~1 receives
  payoff~1 we must have that $q_k(x)=0$.
\end{corollary}
\begin{proof}
  Again, for Player~1 to receive payoff~1, neither of Player~6 and Player~7
  execute their threats to leave directly to the the terminal nodes. For this to
  happen it is required that $\tfrac{1}{2n^2}(\norm{x}_1^2+q_k(x)) \geq
  \tfrac{1}{2n^2}$ and $\tfrac{1}{2n^2}(\norm{x}_1^2-q_k(x)) \geq
  \tfrac{1}{2n^2}$. When $\norm{x}_1=1$ this implies that
  $\tfrac{1}{2n^2}q_k(x)\geq 0$ and $-\tfrac{1}{2n^2}q_k(x)\geq 0$, and thus
  $q_k(x)=0$.
\end{proof}

\begin{figure}[ht]
  \centering
    \begin{tikzpicture}[shorten >= 1, node distance = 2cm, on grid, minimum size = 0.8cm]
    \node[shape=diamond, draw=black, label=above:$\downarrow$] (e1) {$v_0$};
    \node[draw=none] (G1) [left=2cm of e1] {$\Gvar$};
    \node[shape=diamond, draw=black] (e2) [right=2cm of e1] {};
    \node[draw=none] (q1) [above right=0.8cm and 2cm of e2] {$\Gpoly(1)$};
    \node[draw=none] (qi) [right=of e2] {\phantom{MM}};
    \node[draw=none] (qn) [below right=0.8cm and 2cm of e2] {$\Gpoly(\ell)$};

    \path[->]
      (e1) edge node [below] {$\frac{1}{2}$} (G1)
      (e1) edge node [below] {$\frac{1}{2}$} (e2)
      (e2) edge [bend left=25] node [above left=-0.2cm] {$\frac{1}{\ell}$} (q1)
      (e2) edge [bend right=25] node [below left=-0.2cm] {$\frac{1}{\ell}$} (qn)
    ;
    \path[->] (e2) edge (qi) [loosely dashed];
    \path (q1) edge node [black, opacity=1, sloped] {\dots} (qn) [opacity=0];

  \end{tikzpicture}
  \caption{The game $\GS$.}
  \label{fig:G(S)}
\end{figure}

We now have all the ingredients needed for our \cETR-hardness proof.
\begin{proof}[Proof of Theorem~\ref{THM:RecursiveGamesETR}]
  We already discussed the proof of \cETR-membership. For proving \cETR\
  hardness we reduce from $\HOMQUAD(\Simplex)$. As above, let $\calS$ be a
  system of homogeneous quadratic polynomials $q_1(x),\dots,q_\ell(x)$ in
  variables $x=(x_1,\dots,x_n)$. We construct the game $\GS$ as shown in
  Figure~\ref{fig:G(S)}. Using initial chance nodes, play proceeds to $\Gvar$
  with probability~$\tfrac{1}{2}$ and to $\Gpoly(k)$ with probability
  $\tfrac{1}{2\ell}$, for $k=1,\dots,\ell$.

  We shall prove that $\GS$ has a stationary NE satisfying the payoff demands
  \[
    L=\left(
      \frac{1}{2},
      \frac{1}{n},
      1-\frac{1}{n},
      \frac{1 + n}{2n^2},
      \frac{n^2-1}{2n^2},
      \frac{1}{4n^2},
      \frac{1}{4n^2}
    \right) \enspace ,
  \]
  if and only if there exists $x \in \Simplex^{n-1}$ such that $q_k(x)=0$, for
  all $k$.

  Suppose first that $\GS$ has a NE satisfying the payoff demands~$L$. Since
  Player~1 receives payoff~0 in $\Gvar$, Player~1 must receive payoff~1 in every
  game $\Gpoly(k)$. Thus by Lemma~\ref{LEM:G_poly-NE} Player~2 and Player~3
  receive payoff~$\tfrac{1}{n}\norm{x}_1$ and $1-\tfrac{1}{n}\norm{x}_{1}$,
  respectively, which by Lemma~\ref{LEM:G_x-payoff} also is their payoff in
  $\Gvar$. We conclude that $\tfrac{1}{n}\norm{x}_{1}$ and
  $1-\tfrac{1}{n}\norm{x}_{1}$ is also the payoff of Player~2 and Player~3 in
  $\GS$. The payoff demands $L$ gives that $\tfrac{1}{n}\norm{x}_1 \geq
  \tfrac{1}{n}$ and $1-\tfrac{1}{n}\norm{x}_{1} \geq 1-\tfrac{1}{n}$, which
  implies $\norm{x}_1=1$. By Corollary~\ref{COR:G_poly-NE} this implies
  $q_k(x)=0$ for all~$k$.

  Suppose now that $x \in \Simplex^{n-1}$ is such that $q_k(x)=0$ for all~$k$.
  We let Player~1 play according to $x$ in $\Gvar$ and consistent to that (i.e.\
  also according to $x$) in $\Gmul(i,j,(1+a_{ij}^k)/2)$, for all $i,j,k$. We let
  all other players \emph{not} execute any of their threats. It remains to be
  shown that this strategy profile $\tau$ is a NE. No strategy profile yields
  payoff larger than~$\tfrac{1}{2}$ to Player~1, so Player~1 has no incentive to
  change strategy. What remains to prove is that no player gains from executing
  a threat. In $\Gmul(i,j,(1+a_{ij}^k)/2)$, if either Player~2 or 3 execute
  their threat to $v_i$ in $\Gvar$ then their payoff stays unchanged, since
  Player~1 is playing according to $x_i$ in both $v_i$ and $w_3$. Likewise, the
  payoffs for Player~4 and Player~5 are neither improved by executing their
  threat to $v_j$. 
  In $\Gpoly(k)$, since $\norm{x}_1=1$ and $q_k(x)=0$, Player~6 and Player~7 are
  both receiving payoff~$\tfrac{1}{2n^2}$ which is also exactly what they would
  receive by executing their threat. This concludes the proof that $x$ defines a
  NE. Let us finally note that the payoff profile of $\Gvar$ is
  $(0,\tfrac{1}{n},1-\tfrac{1}{n},\tfrac{1}{n},1-\tfrac{1}{n},0,0)$ and the
  (average of) the payoff profiles of $\Gpoly(k)$ is $(1,\tfrac{1}{n},$
  $1-\tfrac{1}{n},\tfrac{1}{n^2},\tfrac{1}{n}(1-\tfrac{1}{n}),\tfrac{1}{2n^2},\tfrac{1}{2n^2})$,
  and the average of these is exactly~$L$. Let us finally note that $\tau$ is
  easily seen to in fact be a SPE.
\end{proof}
\begin{remark} \label{REM:First3Entries}
  We only used the first~$3$ entries of the payoff demands~$L$ to argue a NE
  satisfying the payoff demand implies the system $\calS$ is satisfied. We could
  therefore equivalently have used the demands
  $L=(\tfrac{1}{2},\tfrac{1}{n},\tfrac{n-1}{n},0,0,0,0)$.
\end{remark}

\subsection{Deterministic Games} \label{SEC:DeterministicGames}
Ummels and Wojtczak~\cite{CONCUR:UmmelsW2011} constructed a gadget that allows
for simulation of a chance node by a deterministic game under certain
conditions. Ummels and Wojtczak used this to prove that deciding existence of a
stationary NE is $\SqrtSum$-hard for 8-player recursive games. Their proof
constructs games with both positive and negative terminal rewards. Terminals
with negative rewards are used to make a player prefer infinite play away from
terminals to such a terminal. We describe their gadget below, modified to have
non-negative terminal rewards (and thus not applicable in the reduction of
Ummels and Wojtczak). In acyclic games, as we have constructed, any play reaches
a terminal, and in turn makes non-negative rewards sufficient.

Let $p \in \CornerSimplex^n$ with $\norm{p}_1 < 1$. We construct a gadget game
$\Gchance(p)$ with designated nodes $u_1\dots,u_n$ in order to simulate a single
chance node that for each $i=1,\dots,n$ continues play in nodes~$u_i$ with
probability~$p_i$ and with the remaining probability~${1-\norm{p}_1>0}$ leads to
a terminal~$\bot$.

Define $q_1,\dots,q_n$ by
\[
  q_i = \frac{1-\sum_{j=i}^n p_j}{1-\sum_{j=i+1}^{n}p_j} \enspace .
\]
Note that $\prod_{j=i}^n q_j = 1 - \sum_{j=i}^n p_j$ for all $i = 1, \dots, n$.
The chance node described above can be simulated by the following stochastic
process in steps $k=0,\dots,n$. When $k<n$, we select node $u_{n-k}$ as the
outcome with probability $1-q_{n-k}$, and otherwise proceed to the next step
$k+1$ with probability $q_{n-k}$. When $k=n$, we end with outcome~$\bot$. Then
the probability of outcome $u_i$ is equal to
\begin{equation*}
  (1-q_i)\prod_{j=i+1}^n q_j
  = (\prod_{j=i+1}^n q_j) - (\prod_{j=i}^n q_j)
  = (1 - \sum_{j=i+1}^n p_j) - (1 - \sum_{j=i}^n p_j) =
  p_i
\end{equation*}
as required.

\begin{figure}[ht]
  \centering
  \begin{tikzpicture}[shorten >= 1, node distance = 2cm, on grid, minimum size = 0.8cm]
  
  \node[shape=circle, draw=black, label=above:$2$, label=left:$\rightarrow$] (sn) {$s_n$};

  \node[shape=circle, draw=black, label=above:$1$] (ti1) [right = 2.1cm of sn] {$t_{i+1}$};
  \node[shape=circle, draw=black, label=above:$2$] (si) [right  = 1.5cm of ti1] {$s_{i}$};
  \node[shape=circle, draw=black, label=above:$3$] (ri) [right  = 1.5cm of si] {$r_{i}$};
  \node[shape=circle, draw=black, label=above:$1$] (ti) [right  = 1.5cm of ri] {$t_{i}$};

  \node[shape=circle, draw=black, label=above:$1$] (t1) [right  = 2.1cm of ti] {$t_{1}$};
  \node[draw=none] (t1_term) [right  = 1.5cm of t1] {$(1,0,1)$};

  \node[draw=none] (dots_ni) [right  = 1cm of sn] {$\cdots$};
  \node[draw=none] (dots_i1) [right  = 1cm of ti] {$\cdots$};
  \path[-]
    (sn) edge (dots_ni)
    (ti) edge (dots_i1)
  ;
  \path[->]
    (dots_ni) edge (ti1)
    (dots_i1) edge (t1)

    (ti1) edge node[above=-.1cm] {\small \color{gray} $q_{i+1}'$} (si)
    (si) edge (ri)
    (ri) edge (ti)

    (t1) edge node[above=-.1cm] {\small \color{gray} $q_1'$} (t1_term)
  ;

  \node[draw=none] (sn_term) [below = 1.5cm of sn] {$(0,1-\hat{q}_n,0)$};
  \node[shape=circle, draw=black] (ti1_term) [below = 1.5cm of ti1] {$u_{i+1}$};
  \node[draw=none] (si_term) [below = 1.5cm of si] {$(0,1-\hat{q}_i,0)$};
  \node[draw=none] (ri_term) [below = 1.5cm of ri] {$(0,0,\hat{q}_i)$};
  \node[shape=circle, draw=black] (ti_term) [below = 1.5cm of ti] {$u_{i}$};
  \node[shape=circle, draw=black] (t1_term2) [below = 1.5cm of t1] {$u_{1}$};

  \path[->]
    (sn) edge (sn_term)
    (ti1) edge node[right] {\small \color{gray} $1 - q_{i+1}'$} (ti1_term)
    (si) edge (si_term)
    (ri) edge (ri_term)
    (ti) edge node[right] {\small \color{gray} $1 - q_i'$} (ti_term)
    (t1) edge node[right] {\small \color{gray} $1 - q_1'$} (t1_term2)
  ;
\end{tikzpicture}

  \caption{The game $\Gchance(p)$.}
  \label{fig:Gchance}
\end{figure}

The game $\Gchance(p)$ shown in Figure~\ref{fig:Gchance} consists of
non-terminal nodes $s_i,t_i,r_i$ and $u_i$, for $i=1,\dots,n$, with the initial
node being $s_n$. Player~1 has the role of implementing the chance node, whereas
Player~2 and Player~3 incentivize Player~1 to play using the probabilities
$q_1,\dots,q_n$ by means of threats. In nodes~$t_i$ Player~1 has the choice
between node~$u_i$, or when $i>1$ continuing in node~$s_{i-1}$ and when $i=1$
end in a terminal with rewards~$(1,0,1)$, corresponding to $\bot$. Before each
node~$t_i$, Player~2 and Player~3 are able to threat to end in terminals with
rewards~$(0,1-\hat{q}_i,0)$ and $(0,0,\hat{q}_i)$ from nodes $s_i$ and $t_i$,
respectively, where we define $\hat{q}_i$ by
\[
  \hat{q}_i = \prod_{j=1}^i q_j = \frac{1-\sum_{j=1}^n
    p_j}{1-\sum_{j=i+1}^n p_j} \enspace .
\]

\begin{lemma} \label{LEM:Gchance}
  Consider the game derived from $\Gchance(p)$ where each node $u_i$ is changed
  to be a terminal node with rewards~$(1,1,0)$. Then, play according to any
  stationary NE in which Player~1 receives payoff~1 reaches terminal $u_i$ with
  probability~$p_i$, for all~$i$.
\end{lemma}
\begin{proof}
  For Player~1 to receive payoff~1, play must reach either one of the
  terminals~$u_i$ or~$\bot$ with probability~1, so no threat is executed by
  Player~2 and Player~3. Suppose Player~1 chooses node $u_i$ with probability
  $1-q'_i$, for every $i$. Since Player~3 only receives a positive reward
  in~$\bot$, play must reach $\bot$ with positive probability which means
  $q'_i>0$ for all~$i$. For a given~$i$ and conditioned on play reaching $s_i$,
  Player~2 receives payoff~$1-\prod_{j=1}^i q'_j$ and Player~3 receives payoff
  $\prod_{j=1}^i q'_j$. For Player~2 and Player~3 to not execute their threats
  in $s_i$ and $r_i$ it is required that $1-\prod_{j=1}^i q'_j \geq
  1-\prod_{j=1}^i q_j$ and $\prod_{j=1}^i q'_j \geq \prod_{j=1}^i q_j$, which
  implies $\prod_{j=1}^i q'_j = \prod_{j=1}^i q_j$. Since this must hold for
  all~$i$, we have $q'_i=q_i$ for all~$i$, and thus play reaches terminal $u_i$
  with probability~$p_i$ for all $i$.
\end{proof}

Using the construction above, we are able to replace the chance nodes in $\GS$
used to prove Theorem \ref{THM:RecursiveGamesETR}. The chance node $v_0$ and its
two immediate chance nodes can be combined into a single one with outgoing arcs
to $v_1, v_2, \dots, v_n$ in $\Gvar$ with probability $\tfrac{1}{4n}$ and arcs
to $\Gpoly(1), \dots, \Gpoly(\ell)$ with probability $\tfrac{1}{4\ell}$.
With the remaining probability of $\tfrac{1}{2}$ the chance node leads to a new
terminal $\bot_0$ where all $7$ original players of $\GS$ receive payoff $0$.
This modified chance node can be replaced with the gadget of
Lemma~\ref{LEM:Gchance}, which adds three new players to the construction. In
the terminals of all subgames (including the terminals added next), the first
two newly added players receive payoff $1$ while the third receives $0$.
Similarly, the chance node within $\Gpoly(k)$ can be replaced with a chance
node, that leads with probability $\tfrac{1}{2}$ to a terminal $\bot_k$ and with
probability $\tfrac{1}{2n^2}$ to $\Gmul(x_i,x_j, (1+a^{k}_{i,j})/2)$, for all
$i,j,k$. To compensate for this new terminal, the payoff in the threats by the
original sixth and seventh player is decreased to $\tfrac{1}{4n^2}$. Since each
$\Gpoly(k)$ is independent of another, these chance nodes can be replaced by
only adding another $3$ players, rather than $3\ell$. The first two of these
players gain payoff $1$ in all $\Gmul(i,j, (1+a^{k}_{i,j})/2)$ and the last
gains payoff $0$, while all three gain $0$ in the $\bot_0$ and in $\Gvar$.

We therefore obtain the following result for deterministic recursive games.

\begin{theorem} \label{THM:RecursiveGamesETR_withoutChances}
  It is \cETR-complete to decide whether for a given $m$-player deterministic
  recursive game $G$ and payoff demands $L\in\RR^m$ there exists a stationary NE
  $\tau$ with $U(\tau)\geq L$. The problem is \cETR-complete even for
  13-player acyclic deterministic recursive games with non-negative rewards. The
  same result holds for the analogous problem for stationary SPE.
\end{theorem}
\begin{proof}
  The result follows by similar argumentation as in the proof of
  Theorem~\ref{THM:RecursiveGamesETR} on the payoff vector
  $L = (
    \tfrac{1}{8},
    \tfrac{3}{8n},
    \tfrac{3}{8}(1-\tfrac{1}{n}),
    0,
    0,
    0,
    0,
    1,
    0,
    0,
    \tfrac{1}{4},
    0,
    0
  )$ together with Lemma~\ref{LEM:Gchance}.
\end{proof}

\subsection{Stationary NE where a Player Wins Almost Surely}
Theorem~\ref{THM:RecursiveGamesETR} is concerned with the existence of a
Stationary NE given a payoff demand, and was proven using a payoff demand $L$
that is non-zero for more than one player on a game \GS\ with fractional rewards
in $[0,1]$. In applications of verification and synthesis it is of interest to
discern whether there exists a Nash equilibria, where a single player can expect
payoff $1$ in a game with boolean terminal rewards in $\{0,1\}$; that is, where
a player is \emph{almost surely winning}.

Membership in \cETR\ still follows by the same argumentation as at the beginning
of Section~\ref{SEC:ETR-Completeness}, since the payoff demands is $1$ for the
one player of interest and zero for all others.

\begin{figure}[ht!]
  \centering
  \begin{tikzpicture}[shorten >= 1, node distance = 2cm, on grid, minimum size = 0.8cm]
  
  \node[shape=circle, draw=black, label=above:$1$, label=left:$\rightarrow$] (threat1) {$t_1$};
  \node[shape=circle, draw=black, label=above:$2$] (threat2) [right  = 1.6cm of threat1] {$t_2$};
  \node[shape=circle, draw=black, label=above:$3$] (threat3) [right  = 1.6cm of threat2] {$t_3$};

  \node[draw=none] (game_reduction) [right = 2cm of threat3] {$\GS$};
  \node[draw=none] (game_no_NE) [below  = 1.3cm of threat2] {$\left( \frac{1}{2}, \frac{1}{n}, \frac{n-1}{n}, 0, 0, 0, 0, 0 \right)$};

  \draw[->]
    (threat1) edge (threat2)
              edge [bend left=20] (game_no_NE)
    (threat2) edge (threat3)
              edge (game_no_NE)
    (threat3) edge (game_reduction)
              edge [bend right=20] (game_no_NE)
  ;
\end{tikzpicture}

  \caption{The game $\Gsure(\calS)$.}
  \label{fig:Gsure}
\end{figure}

Consider the game $\Gsure(\calS)$ in Figure~\ref{fig:Gsure}, where Player~1,
Player~2, and Player~3 can choose to not continue into the game $\GS$ used in
the proof of Theorem~\ref{THM:RecursiveGamesETR}, but instead end the game early
at a terminal with payoff $L$ of Remark~\ref{REM:First3Entries}. An eighth
player is added, who always gains payoff $1$ in $\GS$ but only payoff $0$ at the
newly added terminal. Since this construction only consists of non-negative
fractional terminal rewards, then one may replace all terminals with chance
nodes that lead to terminal rewards in $\{0,1\}$ without altering the expected
payoff. We then obtain the following extension of
Theorem~\ref{THM:RecursiveGamesETR}.

\begin{theorem} \label{THM:RecursiveGamesETR_surelyWinning}
  It is \cETR-complete to decide whether for a given $m$-player recursive game
  $G$, in which all rewards are~0 or~1, and a given~$k$, there exists a
  stationary NE in which Player~$k$ is almost surely winning. The problem is
  \cETR-complete even for acyclic 8-player recursive games. The same result
  holds for the analogous problem for stationary SPE.
\end{theorem}
\begin{proof}
  In a stationary NE of $\Gsure(\calS)$, Player~8 is almost surely winning if
  and only if neither Player~1, Player~2, nor Player~3 execute their threat to
  end the game early. Their expected payoff for executing the threat is
  respectively $\tfrac{1}{2}$, $\tfrac{1}{n}$, and $\tfrac{n-1}{n}$, which also
  is their expected payoff in the proof of Theorem~\ref{THM:RecursiveGamesETR}.
  That is, the three players effectively enforce the payoff demand $L$ of
  Remark~\ref{REM:First3Entries} to $\GS$.
\end{proof}

\subsection{Stationary NE without Payoff Demands}
Theorem~\ref{THM:RecursiveGamesETR} settles the complexity of deciding existence
of stationary NE satisfying payoff demands. While deciding the mere existence of
any stationary NE may seem to be an easier problem, we show that is exactly as
hard.

\cETR-membership again follows as in Section~\ref{SEC:ETR-Completeness} by the
proof of Ummels and Wojtczak for
\PSPACE-membership~\cite[Theorem~7]{CONCUR:UmmelsW2011}. One will notice, that
the proof in \cite{arXiv:UmmelsW2011} has the payoff demand expressed in a
separate clause. One may omit this to the only express existence of a stationary
NE.

Suppose that we have an $m$-player gadget game $\GnoNE$ which does not have a
stationary NE and Player~1, 2, and 3 receive payoff~$0$ for \emph{all} strategy
profiles of the players; examples of such gadgets will be elaborated below. Let
$L$ be given by Remark~\ref{REM:First3Entries}. Construct now the game
$\GexistsNE(\calS)$ shown in Figure~\ref{fig:GexistsNE}, where similar to
$\Gsure(\calS)$ the first three players can choose not to go into $\GS$ but into
a different subgame. On this alternative path, a chance node $t_4$ leads with
probability $\tfrac{1}{2}$ to a terminal, which has twice the payoff demand $L$.
In the other case, the chance node leads to $\GnoNE$.

\begin{figure}[ht]
  \centering

  \begin{tikzpicture}[shorten >= 1, node distance = 2cm, on grid, minimum size = 0.8cm]
  
  \node[shape=circle, draw=black, label=above:$1$, label=left:$\rightarrow$] (threat1) {$t_1$};
  \node[shape=circle, draw=black, label=above:$2$] (threat2) [right  = 1.6cm of threat1] {$t_2$};
  \node[shape=circle, draw=black, label=above:$3$] (threat3) [right  = 1.6cm of threat2] {$t_3$};

  \node[shape=diamond, draw=black] (chance) [below  = 1.2cm of threat2] {$t_4$};

  \node[draw=none] (game_reduction) [right  = 2cm of threat3] {$\GS$};
  \node[draw=none] (game_no_NE) [left  = 1.7cm of chance] {$\GnoNE$};
  \node[draw=none] (chance_term) [right  = 2.8cm of chance] {$2 \cdot \left( \frac{1}{2}, \frac{1}{n}, \frac{n-1}{n}, 0, \dots, 0 \right)$};

  \draw[->]
    (threat1) edge (threat2)
              edge [bend left=20] (chance)
    (threat2) edge (threat3)
              edge (chance)
    (threat3) edge (game_reduction)
              edge [bend right=20] (chance)

    (chance) edge node[below] {$\frac{1}{2}$} (game_no_NE)
             edge node[below] {$\frac{1}{2}$} (chance_term)
  ;
\end{tikzpicture}


  \caption{The game $\GexistsNE(\calS)$}
  \label{fig:GexistsNE}
\end{figure}

\begin{lemma} \label{LEM:GexistsNE}
  The game $\GexistsNE(\calS)$ has a stationary NE if and only if there is a
  stationary NE of $\GS$ satisfying the payoff demands~$L$.
\end{lemma}
\begin{proof}
  Since $\GnoNE$ does not permit a stationary NE, the game $\GexistsNE(\calS)$
  has a NE if and only if none of Player~1, Player~2, and Player~3 execute their
  threat to go to $t_4$. Similarly to the proof of
  Theorem~\ref{THM:RecursiveGamesETR_surelyWinning}, the three players enforce
  the payoff demand $L$ of Remark~\ref{REM:First3Entries} to $\GS$.
\end{proof}

Boros and Gurvich~\cite{MSS:BorosG2003} and
Kuipers~et~al.~\cite{EJOR:KuipersFSV2009} (cf.\ \cite[Proposition
3.3]{Ummels2011}) construct a (cyclic) 3-player recursive game with non-negative
rewards which has no stationary NE. We may let Player~4, 5, and~6 take the role
of playing in this game, letting Player~1, 2, and~3 receive reward~0 in all
terminals. Together with Theorem~\ref{THM:RecursiveGamesETR} we obtain the
following result.

\begin{theorem} \label{THM:RecursiveGamesETR_withoutPayoff}
  It is \cETR-complete to decide whether a given $m$-player recursive game has
  a stationary NE, even for 7-player recursive games with non-negative rewards.
\end{theorem}

\begin{figure}[ht]
  \centering

  \begin{tikzpicture}[shorten >= 1, node distance = 2cm, on grid, minimum size = 0.8cm]
  
  \node[shape=circle, draw=black, label=above:$1$, label=left:$\rightarrow$] (threat1)
    {$t_1$};
  \node[shape=circle, draw=black, label=above:$2$] (threat2)
    [above left = 1.4cm and 0.5cm of threat1] {$t_2$};
  \node[shape=circle, draw=black, label=above:$3$] (threat3)
    [above right = 1cm and 1.2cm of threat2] {$t_3$};
  \node[shape=circle, draw=black, label=above:$8$] (threat8)
    [below right = 1cm and 1.2cm of threat3] {$t_8$};
  \node[shape=circle, draw=black, label=above:$11$] (threat11)
    [below left = 1.4cm and 0.5cm of threat8] {$t_{11}$};

  \node[draw=none] (game_reduction) [right = 1.5cm of threat11] {$\GS'$};
  \node[draw=none] (game_no_NE) [below  = 1.3cm of threat3] {$\GnoNE$};

  \draw[->]
    (threat1)  edge[bend left=10] (threat2)
               edge (game_no_NE)
    (threat2)  edge[bend left=10] (threat3)
               edge (game_no_NE)
    (threat3)  edge[bend left=10] (threat8)
               edge (game_no_NE)
    (threat8)  edge[bend left=10] (threat11)
               edge (game_no_NE)
    (threat11) edge (game_reduction)
               edge (game_no_NE)
  ;
\end{tikzpicture}


  \caption{The game $\GexistsNE'(\calS)$}
  \label{fig:GexistsNE'}
\end{figure}

In continuation of Section~\ref{SEC:DeterministicGames} we would like to
dispense with the chance node $t_4$ to thereby combine
Theorem~\ref{THM:RecursiveGamesETR_withoutChances} with
Theorem~\ref{THM:RecursiveGamesETR_withoutPayoff}. We thus consider the game in
Figure~\ref{fig:GexistsNE'}, where $t_4$ has been removed in favor of going
directly to $\GnoNE$ and a threat is added for the chance node implemented by
Player~$8$ and $11$. Unlike above, since play never reaching a terminal results
in payoff~0, then it is not possible to guarantee positive payoffs in the
$\GnoNE$. Instead, let $\GS'$ be the game obtained from $\GS$ of
Theorem~\ref{THM:RecursiveGamesETR_withoutChances} where all terminal rewards of
Player~1, 2, 3, 8, and~11 have been decreased by $\tfrac{1}{8}, \tfrac{3}{8n},
\tfrac{3n-3}{8n}, 1,$ and $\tfrac{1}{4}$, respectively. Since the game is
acyclic, and hence reaches a terminal with probability~1, this does not change
the NE of the game, but just subtracts $\tfrac{1}{8}, \tfrac{3}{8n},
\tfrac{3n-3}{8n}, 1,$ and $\tfrac{1}{4}$ from the NE payoffs of Player~1, 2, 3,
8, and~11, respectively.

Boros~et~al.~\cite{DAM:BorosGMOV18} recently constructed a \emph{deterministic}
3-player recursive game without a stationary NE. As above, we can let Player~4,
5, and~6 take the role of playing in this game, letting Player~1, 2, and~3
receive reward~0 in all terminals. Repeating the arguments in the proof of
Lemma~\ref{LEM:GexistsNE} and
Theorem~\ref{THM:RecursiveGamesETR_withoutChances} we obtain the following
result.
\begin{theorem} \label{THM:RecursiveGamesETR_withoutPayoff_withoutChances}
  It is \cETR-complete to decide whether a given $m$-player deterministic
  recursive game has a stationary NE, even for $m=13$.
\end{theorem}

\subsection{\texorpdfstring{$\omega$}{ω}-Regular Objectives and Mean-Payoff Games}
\label{SEC:BorelMeanPayoff}
Ummels and Wojtczak proved membership of \PSPACE\ by giving reductions to \ETR\
for the problem of deciding existence of a stationary NE meeting given payoff
constraints in several classes of games. For perfect information stochastic games
where all players have Streett or Rabin objectives (called Streett-Rabin games),
or where all players have Muller objectives a non-determistic polynomial time
many-one reduction to \ETR\ is constructed~\cite{LMCS:UmmelsW2011}. For perfect
information mean-payoff games a (deterministic) polynomial time many-one reduction
to \ETR\ is constructed~\cite{CONCUR:UmmelsW2011}. Using the characterization of
\cETR\ in terms of nondeterministic Blum-Shub-Smale machines it is straightforward
that the many-one reductions may be combined with decision of $\ETR$, thereby
proving \cETR-membership for the problems of deciding existence of NE meeting given
payoff constraints.

Street-Rabin games generalize reach-a-set games and stay-in-a-set games where
all objectives are terminal. One may prove \cETR-membership for general
reach-a-set and stay-in-a-set games in a similar way as Ummels and Wojtczak
did.
\begin{theorem}
  It is \cETR-complete to decide whether a given $m$-player perfect
  information reach-a-set game has a stationary NE, even for $m=7$.
\end{theorem}
\begin{proof}
  Recursive games with non-negative rewards may, after normalizing rewards to
  $[0,1]$, be viewed as a special case of reach-a-set games. The result then
  follows from Theorem~\ref{THM:RecursiveGamesETR_withoutPayoff}.
\end{proof}

\begin{theorem}
  It is \cETR-complete to decide whether a given $m$-player perfect
  information stay-in-a-set game has a stationary NE, even for $m=7$.
\end{theorem}
\begin{proof}
  Hansen and Raskin~\cite{GandALF:HansenR2019} constructed a 2-player perfect
  information stay-in-a-set game without any stationary NE. We may use this game
  in place of $\GnoNE$ in the proof of
  Theorem~\ref{THM:RecursiveGamesETR_withoutPayoff}. Namely, consider
  transforming the game $\GexistsNE(\calS)$ by first dividing all rewards by~2
  and then subtracting~1 from all rewards. This does not alter the set of NE of
  the game, but maps all rewards to the interval $[-1,0]$, which may then be
  viewed as a stay-in-a-set game with terminal safety objectives. We may then
  replace $\GnoNE$ by the 2-player stay-in-a-set game of Hansen and Raskin,
  where we let Player~4 and Player~5 take the role of the 2~players and all
  nodes of this game are excluded from the safe sets of Player~1, 2, and~3.
\end{proof}

Let us finally consider mean-payoff games. Ummels and
Wojtczak~\cite{CONCUR:UmmelsW2011} note that non-negative fractional terminal
rewards may in mean-payoff games be simulated with a simple cycle where all
rewards are chosen from the set $\{0,1\}$. Since \GS\ used to prove
Theorem~\ref{THM:RecursiveGamesETR} and
\ref{THM:RecursiveGamesETR_withoutChances} only has non-negative fractional
terminal rewards, so with the \cETR-membership result above we thus obtain
analogous results to Theorem~\ref{THM:RecursiveGamesETR} and
\ref{THM:RecursiveGamesETR_withoutChances} for mean-payoff games where all
rewards are~0 or~1.

\begin{theorem}
  It is \cETR-complete to decide whether a given $m$-player perfect information
  mean-payoff game where all rewards are~0 or~1 has a stationary NE that
  satisfies a given payoff demand, even for $m=7$. The same result holds for the
  analogous problem for stationary SPE.
\end{theorem}

\begin{theorem}
  It is \cETR-complete to decide whether a given $m$-player deterministic
  perfect information mean-payoff game where all rewards are~0 or~1 has a
  stationary NE that satisfies a given payoff demand, even for $m=13$. The same
  result holds for the analogous problem for stationary SPE.
\end{theorem}

\section{Equilibria in Tree Games}
\label{SEC:TreeGames}
Littman~et~al.~\cite{UAI:LittmanRTZ2006} proved it \NP-hard to decide existence
of a NE satisfying given payoff demands for 2-player game trees. The proof
applies both to positional and stationary NE, as well as positional and
stationary SPE. We describe a variation of their proof below for completeness.

Recall that \PARTITION\ is the \NP-complete problem of deciding
whether for given $a \in \ZZ_+^n$ there exist
$S \subseteq \{1,\dots,n\}$ such that $\sum_{i\in S}a_i=K/2$, where
$K=\sum_{i=1}^na_i$. For a rational valued vector $a \in \QQ_+^n$,
which is possible in our constructions, the problem is even strongly
\NP-complete~\cite{CSR:Wojtczak2018}. We define $\Gpart(a)$ to be the
game depicted in Figure~\ref{fig:partition}, where an initial chance
node selects an item $i$ uniformly at random, Player~1 can then choose
to give the item to Player~2 (i.e.\ choose rewards $(0,a_i)$), or pass
the turn to Player~2, who may either give the item to Player~1 (i.e.\
choose rewards $(a_i,0)$ or discard it (i.e. choose rewards $(0,0)$).
\begin{figure}[ht]
  \centering
  \begin{tikzpicture}[shorten >= 1, node distance = 2cm, on grid, minimum size = 0.8cm]
  \node[shape=diamond, draw=black, label=left:$\rightarrow$] (e1) {};
  \node[draw=none] (n1_i) [right=of e1] {\phantom{MM}};
  \node[shape=circle, draw=black, label=above:$1$] (n1_1) [above=1.5cm of n1_i] {$u_1$};
  \node[shape=circle, draw=black, label=above:$2$] (n2_1) [above right=0.5cm and 2cm of n1_1] {$v_1$};  
  \node[draw=none] (g1_1) [above right=0.5cm and 2cm of n2_1] {$(a_1,0)$};
  \node[draw=none] (g1_2) [below right=0.5cm and 2cm of n2_1] {$(0,0)$};
  \node[draw=none] (g1_3) [below right=0.5cm and 2cm of n1_1] {$(0,a_1)$};

  \node[shape=circle, draw=black, label=above:$1$] (n1_n) [below=1.5cm of n1_i] {$u_n$};
  \node[shape=circle, draw=black, label=above:$2$] (n2_n) [above right=0.5cm and 2cm of n1_n] {$v_n$};  
  \node[draw=none] (gn_1) [above right=0.5cm and 2cm of n2_n] {$(a_n,0)$};
  \node[draw=none] (gn_2) [below right=0.5cm and 2cm of n2_n] {$(0,0)$};
  \node[draw=none] (gn_3) [below right=0.5cm and 2cm of n1_n] {$(0,a_n)$};

  \path[->]
     (e1) edge [bend left=10] node [above left] {$\frac{1}{n}$} (n1_1)
          edge [bend right=10] node [below left] {$\frac{1}{n}$} (n1_n)
     (n1_1) edge [bend left=8]  (n2_1)
            edge [bend right=8] (g1_3)
     (n2_1) edge [bend left=8] (g1_1)
            edge [bend right=8] (g1_2)
     (n1_n) edge [bend left=8]  (n2_n)
            edge [bend right=8] (gn_3)
     (n2_n) edge [bend left=8] (gn_1)
            edge [bend right=8] (gn_2)
   ;
   \path[->] (e1) edge (n1_i) [loosely dashed];
   \path (n1_1) edge node [black, opacity=1, sloped] {\dots} (n1_n) [opacity=0];
\end{tikzpicture}

  \caption{The partition game $\Gpart(a)$.}
  \label{fig:partition}
\end{figure}
Consider the subgame at node $u_i$, let $p_i$ be the probability that Player~$i$
gives the item to the other player, $i=1,2$. In a NE we have either
$(p_1,p_2)=(1,0)$ or $p_1=0$ (and $p_2$ arbitrary). Thus the only NE (which are
also positional and subgame perfect) where the total payoff of the players is
$a_i$ are $(p_1,p_2)=(1,0)$ and $(p_1,p_2)=(0,1)$. It follows that NE in the
game $\Gpart(a)$ in which both players receive payoff~$K/2$ are positional and
SPE and they correspond exactly to equal partitions of the integers
$a_1,\dots,a_n$. This gives a reduction showing \NP-hardness of deciding
existence of a NE satisfying payoff demands, even for 2-player games.
\begin{theorem}[Littman~et~al.~\cite{UAI:LittmanRTZ2006}]
  It is \NP-hard to decide whether for a given 2-player tree game there exist a
  NE satisfying given payoff demands. This hold for both positional NE,
  stationary NE, positional SPE and stationary SPE.
\end{theorem}
It is not difficult to prove \NP-membership, and with it strongly \NP-completeness
due to \cite{CSR:Wojtczak2018}, for existence of NE (SPE) for 2-player tree
games or for existence of \emph{positional} NE (SPE) for $m$-player tree games,
that satisfies given payoff demands. Proving \NP-membership for existence of
stationary NE (SPE) in $m$-player tree games satisfying given payoff demands is
to our best knowledge yet an open problem.

Using the chance node gadget of Ummels and Wojtczak of Lemma~\ref{LEM:Gchance}
gives strongly \NP-hardness for deterministic tree games.
\begin{corollary}
  It is strongly \NP-hard to decide whether for a given 5-player deterministic
  tree game there exists a stationary NE satisfying given payoff demands. This
  holds also for positional SPE.
\end{corollary}

\section{Conclusion}
In this paper we have focused on the complexity of decision problems concerning
stationary NE in perfect information stochastic games, and we have obtained the
first \cETR-completeness result for such games, even for \emph{acyclic} games.
While existence of NE with payoff constraints is undecidable for general games,
as shown by Ummels and Wojtczak~\cite{LMCS:UmmelsW2011}, it is decidable for
acyclic games. Indeed, for acyclic games, (general) NE correspond to stationary
NE in the unfolding of the game into a tree game. We consider it an interesting
problem to determine the precise complexity of existence of (general) NE meeting
given payoff demands for acyclic games.

\bibliography{SSMGComplexity}

\end{document}